\documentclass[showpacs,preprintnumbers,twocolumn]{revtex4}

\usepackage[centertags]{amsmath}
\usepackage{amsfonts}
\usepackage{amssymb}
\usepackage{amsthm}
\usepackage{newlfont}
\usepackage{graphicx}

\newtheorem{thm}{Theorem}[section]
\newtheorem{lem}[thm]{Lemma}
\newtheorem{conj}[thm]{Conjecture}
\newtheorem{defn}[thm]{Definition}
\newtheorem{cor}[thm]{Corollary}
\newtheorem{obs}[thm]{Observation}
\newtheorem{prop}[thm]{Proposition}

\begin{document}

\title{Symmetric
extendibility of quantum states}
\author{Marcin L. Nowakowski}
\affiliation{Faculty of Applied Physics and Mathematics,
~Gdansk University of Technology, 80-952 Gdansk, Poland}
\affiliation{National Quantum Information Center of Gdansk, Andersa 27, 81-824 Sopot, Poland}


\begin{abstract}
Studies on symmetric extendibility of quantum states become especially important in a context of analysis of one-way quantum measures of
entanglement, distilabillity and security of quantum protocols. In this paper we analyse composite systems containing a symmetric extendible part with a
particular attention devoted to one-way security of such systems. Further, we introduce
a new one-way monotone based on the best symmetric approximation of quantum state. We underpin those results with geometric observations on structures of multi-party settings which posses in sub-spaces substantial symmetric extendible components. Finally, we state a very important conjecture linking symmetric-extendibility with one-way distillability and security of all quantum states analyzing behavior of private key in neighborhood of symmetric extendible states.
\end{abstract}

\maketitle

\section{Introduction}
Recent years have proved a great interest of symmetric extendibility concept showing its usability in quantum communication theory, especially in domain of one-way communication. A natural relation between monogamy of entanglement and symmetric extendibility concept was established \cite{D1,D2,T1} with an important application to analysis of Bell inequalities for multipartite settings where some of the parties possess the same sets of measurement settings. Further, the concept is central for studies of one-way quantum channel capacities, entanglement distillability and private key analysis deriving new upper bounds on these communication rates \cite{MNPH,Lutk1, Lutk2, Lutk3, Lutk4, MNPH2}. It seems also that symmetric extendibility is fundamental for studies on recovery and entanglement breaking channels including its neighborhood \cite{KLi} as well as for such measures like squashed entanglement and quantum discord \cite{MPiani} or analysis of directed communication in 1D/2D spin chains \cite{Renner2}. The aforementioned applications sufficiently prove importance of the notion for
quantum communication theory. The challenge for the present quantum information theory in domain of one-way communication is to better understand behavior of all quantum states in the region of non-symmetric extendibility and in particular in a region of non-positive coherent information \cite{Devetak05} where no known one-way protocol for distillation of entanglement and private key exists. We believe that the following paper will support these studies.
In this paper we provide some new observations about behavior of symmetric states under action of one-way LOCC operations and remind important facts about composite systems containing a symmetric extendible part. Further, we analyze a concept of locking non-symmetric extendibility with its application for security of quantum states asking abut behavior of states assisted by symmetric extendible part. Moreover, we derive a very important link between all two-qubit states not being extendible and one-way entanglement distillation and privacy, analyzing also behavior of private key in neighborhood of symmetric extendible states. In this context, we verified non-symmetric extendible two-qubit Werner states in the region of non-positive coherent information \cite{Winter2} putting an important question about their distillability or existence of one-way bound entanglement.
 We also give a formalized structure to some natural intuitions about nature of composite systems and its reference to k-extendible states.

\section{Symmetric extendible states}
Particularly symmetric extendibility \cite{D1, D2, T1} of a given bipartite state
$\rho_{AB}\in B(\cal{H}_A\otimes\cal{H}_B)$ denotes that there
exists a tripartite state $\rho_{ABE}\in
B(\cal{H}_A\otimes\cal{H}_B\otimes\cal{H}_B)$ invariant due to
permutation of B and E part, namely, if:
\begin{equation}
  P=\sum_{ijk}|ijk\rangle\langle ikj|
\end{equation}
then $P\rho_{ABE}P^{\dagger}=\rho_{ABE}$ and $Tr_{E}\rho_{ABE}=\rho_{AB}=\rho_{AE}$.

By $0$-extendible states we will denote those that are not
symmetrically extendible at all. One could note that it might
be useful to partition the set of all symmetric extendible states $SE$ by relation of
$k$-extendibility. If $\mathcal{S}_{k}$ denotes a convex set \cite{MNPH} of all
states being $k$-extendible, there holds the natural inclusion
relation [Fig. \ref{Fig 1}]:
\begin{equation}\label{incl}
  \mathcal{S}_{1}\supset \mathcal{S}_{2}\supset \ldots \supset
  \mathcal{S}_{k}
\end{equation}

\begin{figure}
\includegraphics[width=9cm, height=5cm]{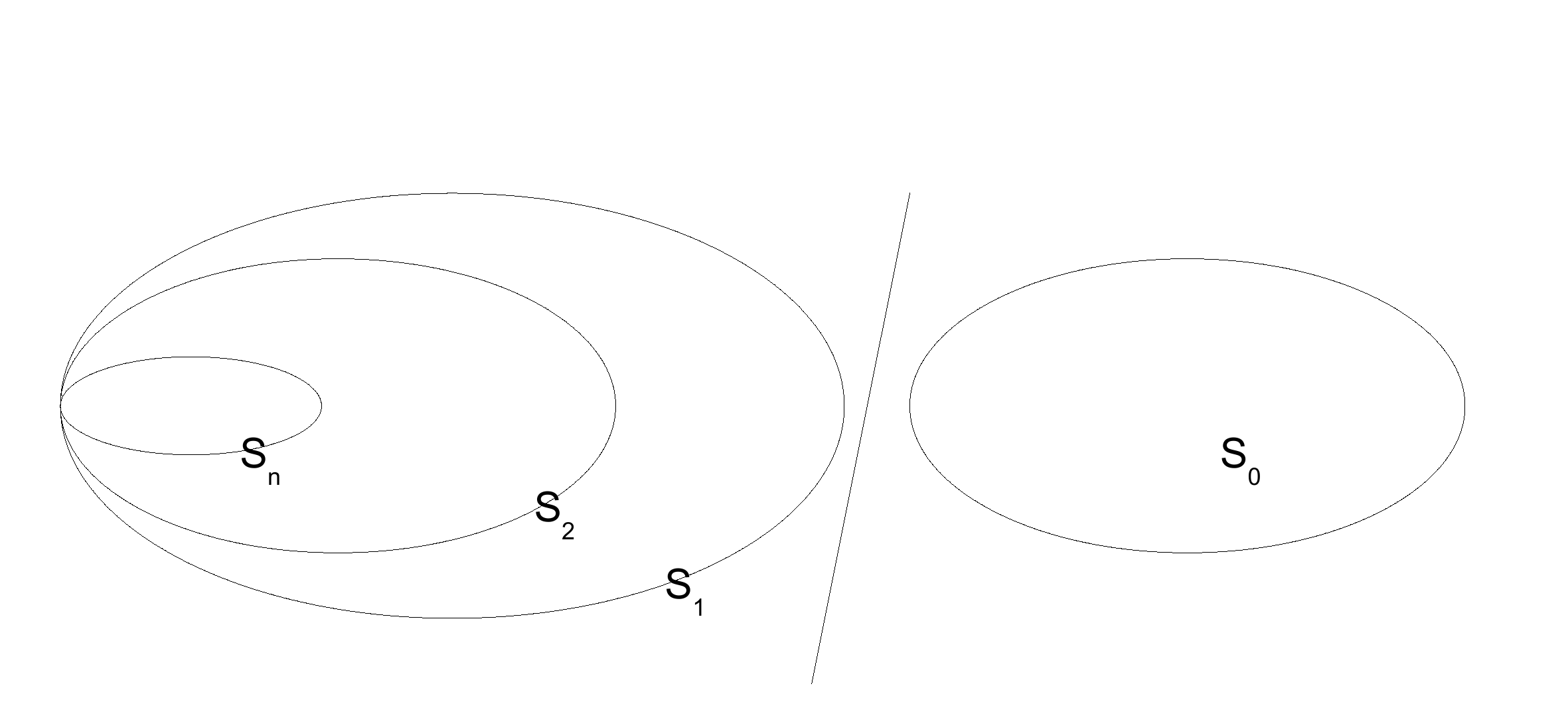}
 \caption[Fig1.]{The space of quantum states can be decomposed by the relation of
 k-extendibility. $\mathcal{S}_{0}$ denotes the set of all non-extendible states whereas
 $\mathcal{S}_{n}$ the set of states having n-rank symmetric extensions.}
  \label{Fig 1}
\end{figure}

Of a great importance is the fact that for a given $\rho_{AB}\in
SE$ there may exist different k-rank symmetric extensions so
that the property is not unique and one could represent the set of
appropriate symmetric extensions by means of equivalence classes
given by the relation $B(\cal{H}_A\otimes\cal{H}_B)\ni
\rho_{AB}\sim \rho \in
B(\cal{H}_A\otimes\cal{H}_B^{\otimes(k+1)})$ if and only if $\rho$
is a k-rank symmetric extension of state $\rho_{AB}$. As the
trivial example note that for
$\rho_{AB}=\frac{1}{2}(|00\rangle\langle 00|+|11\rangle\langle
11|)$ at least the following are extensions of rank one:
$|GHZ\rangle=\frac{1}{\sqrt{2}}(|000\rangle+|111\rangle)$ and
$\rho=\frac{1}{2}(|000\rangle\langle 000| + |111\rangle\langle
111|)$.

For $k$-extendible states it might be useful to introduce an
operator swapping $k+1$ particles:
\begin{equation}
  P_{\pi}=\sum_{i_{1}i_{2}\ldots i_{k+1}}|i_{1}i_{2}\ldots i_{k+1}\rangle
  \langle \pi(i_{1})\pi(i_{2})\ldots \pi(i_{k+1})|
\end{equation}
where swapping is performed for an arbitrary permutation $\pi$.
Hence, there holds a general relation for $k$-extendibility that
explicitly derives set $\mathcal{S}_{k}$: $\forall_{\pi}\;
P_{\pi}\rho_{AB_{1}\ldots B_{k}B_{k+1}
}P_{\pi}^{\dagger}=\rho_{AB_{1}\ldots B_{k}B_{k+1} }$.
\\
\\
\textbf{Example 1.}
As a 1-extendible state we present
$\rho_{AB}=\frac{1}{3}|00\rangle\langle00|+\frac{2}{3}|\Phi_{+}\rangle\langle\Phi_{+}|$
that obviously possess rank-1 symmetric purification to W-state
$|W\rangle=\frac{1}{\sqrt{3}}(|001\rangle+|010\rangle+|100\rangle)$.

We could derive for this example a general form of $n$-extendible
state inheriting from $W$-like n-partite states:
\begin{equation}\label{WState}
  \Upsilon_{AB}(n)=\frac{n}{n+2}|00\rangle\langle
  00|+\frac{2}{n+2}|\Phi_{+}\rangle\langle\Phi_{+}|
\end{equation}
where
$|\Phi_{+}\rangle=\frac{1}{\sqrt{2}}(|01\rangle+|10\rangle)$.
Interestingly one can simply show that for e.g. $GHZ$-like
n-partite states being a maximal extension of
$\rho_{AB}=\frac{1}{2}(|00\rangle\langle00|+|11\rangle\langle11|)$
there holds $\rho_{AB}=\lim_{n\rightarrow\infty}\rho_{AB}(n)$ that
is in agreement with theorems \cite{T1} stating implicitly that
$\rho$ is separable if and only if is $\infty$-extendible (where $\rho_{AB}(n)$ is derived
from n-partite GHZ state).
\\

Following we present two different approaches to the problem of
representation of symmetric extensions in extended space. The
first approach is widely used in previous papers (see \cite{D1,
D2, T1}) on extendibility of quantum states. Every bipartite state
$\rho_{AB}\in B(\mathcal{H}_{A}\otimes \mathcal{H}_{B})$ where
$\mathcal{H}_{A}=\mathbb{C}^{m}$ and
$\mathcal{H}_{B}=\mathbb{C}^{n}$  can be represented in the basis
of generators of group $SU(m)\otimes SU(n)$ as follows:

\begin{eqnarray}\label{rep1}
\rho_{AB}&=&\gamma\sigma_{A}^{0}\otimes \sigma_{B}^{0} +
\sum_{i}\alpha_{i}\sigma_{A}^{0}\otimes\sigma_{B}^{i} + \\
&+& \sum_{j}\beta_{j}\sigma_{A}^{j}\otimes \sigma_{B}^{0} +
\sum_{i,j\neq 0}\zeta_{ij}\sigma_{A}^{i}\otimes \sigma_{B}^{j}
\nonumber
\end{eqnarray}
where $\sigma_{B}^{i}$ are basis elements of $SU(n)$ and
respectively $\sigma_{A}^{i}$ for $SU(m)$.
Elements of the basis satisfy relations:
 $Tr[\sigma_{S}^{i}\sigma_{S}^{j}]=\eta_{S}\delta_{ij}$ and
$Tr[\sigma_{S}^{i}]=\delta_{1i}$  with $S=\{A, B\}$. Therefore, one could derive a
general representation of all 1-rank symmetric extensions:
\begin{eqnarray}
\rho_{AB_{1}B_{2}}&=&\sum_{i,j\neq
0}\alpha_{ij}\sigma_{A}^{i}\otimes \sigma_{B_{1}}^{j}\otimes
\sigma_{B_{2}}^{j} + \\ &+& \sum_{ijk,
j<k}\beta_{ijk}(\sigma_{A}^{i}\otimes\sigma_{B_{1}}^{j}\otimes\sigma_{B_{2}}^{k}
+\sigma_{A}^{i}\otimes\sigma_{B_{1}}^{k}\otimes\sigma_{B_{2}}^{j})
\nonumber
\end{eqnarray}
and further, for general case of $k$-extendibility:

\begin{eqnarray}
& &\rho_{AB_{1}\ldots B_{k+1}}=\sum_{i,j\neq
0}\alpha_{ij}\sigma_{A}^{i}\otimes \sigma_{B_{1}}^{j}\otimes \ldots \otimes\sigma_{B_{k+1}}^{j} +\\
& &\sum_{i,i_{1}<i_{2}<\ldots<
i_{k+1}}\sum_{\sigma}\beta_{ii_{1}\ldots i_{k+1}
}\sigma_{A}^{i}\otimes \sigma_{B_{1}}^{\sigma(i_{1})}\otimes
\ldots \otimes\sigma_{B_{k+1}}^{\sigma(i_{k+1})}\nonumber
\end{eqnarray}

The latter approach that we will utilize in this paper is based on
partitioning a space on which Bobs' states operate into symmetric
and antisymmetric subspace.

Following we will prove some lemmas about Schmidt decomposition of k-rank pure symmetric states that supports in course of the paper more powerful theorem about properties of symmetric extendible states.

\begin{lem}\label{Schmidt1}
Let $\rho_{AB_{1}} \in B(\mathcal{H}_{A}\otimes \mathcal{H}_{B_{1}})$ be symmetrically extendible to a k-rank pure extension $\Psi_{AB_{1}\ldots B_{k+1}} \in \mathcal{H}_{A}\otimes \mathcal{H}^{\otimes k+1}_{B_{1}}$ then there exists a  Schmidt decomposition:
\begin{equation}
\Psi_{AB_{1}\ldots B_{k+1}}=\sum_{i}\alpha_{i}|\phi_{i}^{AB_{1}}\rangle |\psi_{i}^{B_{2\ldots k+1}}\rangle
\end{equation}
where $\{|\phi_{i}^{AB_{1}}\rangle\}, \{|\psi_{i}^{B_{2\ldots k+1}}\rangle\}$ are orthonormal sets and $|\psi_{i}^{B_{2\ldots k+1}}\rangle \in Sym^{k}\bigoplus Asym^{k} (\mathcal{H}_{B_{1}})$.

\end{lem}
\begin{proof}
Since:
\begin{equation*}
\forall_{\pi} \; I_{AB_{1}} \otimes P_{\pi} \Psi_{AB_{1}\ldots B_{k+1}} = \pm \Psi_{AB_{1}\ldots B_{k+1}}
\end{equation*}
where $P_{\pi}$ operates only on $B_{2}\ldots B_{k+1}$ of the system, then $\sum_{i}\alpha_{i}|\phi_{i}^{AB_{1}}\rangle P_{\pi} |\psi_{i}^{B_{2\ldots k+1}}\rangle =\pm \sum_{i}\alpha_{i}|\phi_{i}^{AB_{1}}\rangle|\psi_{i}^{B_{2\ldots k+1}}\rangle$. However, since the state is a symmetric extension, the above Schmidt decomposition is invariant due to any permutation on B-part and $|\phi_{i}^{AB_{1}}\rangle$ indexes uniquely the $|\psi_{i}^{B_{2\ldots k+1}}\rangle$ states so $P_{\pi}$
transforms $|\psi_{i}^{B_{2\ldots k+1}}\rangle$ onto itself. Therefore, the second multiplicands of Schmidt decomposition represent either symmetric or antisymmetric orthonormal states.
\end{proof}

When in \cite{Lutk1} spectral conditions for 1-rank symmetric extensions
were stated, following we derive general statement about spectral conditions for
k-extendible states basing on the observation about decomposition of
symmetric states.

\begin{obs}\label{ObservationSym}
Every pure normalized state $|\Psi\rangle \in Sym^{k+1}\bigoplus Asym^{k+1} (\mathcal{H}_{B_{1}})$ of k+1-partite system can be decomposed to the following Schmidt form:
\[
\forall_{1<l<k} |\Psi\rangle=\sum_{i}|\phi_{i}^{B_{1}\ldots B_{l}}\rangle |\phi_{i}^{B_{l+1}\ldots B_{k+1}}\rangle
\]
where the multiplicands form respectively symmetric or antisymmetric orthonormal sets.
\end{obs}
\begin{proof}
One can conduct the proof similarly to $(\ref{Schmidt1})$. Since $\forall_{\pi} P_{\pi}|\Psi\rangle\langle\Psi|P_{\pi}=|\Psi\rangle\langle\Psi|$, then for all possible permutations the operation
cannot change Schmidt decomposition of $\sum_{i}|\phi_{i}^{B_{1}\ldots B_{l}}\rangle |\phi_{i}^{B_{l+1}\ldots B_{k+1}}\rangle$. Furthermore, due to assumed symmetry property of $|\Psi\rangle$, a state of any l-subsystem $B_{1}\ldots B_{l}$ represented by the first multiplicand is permutationally invariant and the same is applied to the second multiplicand.
\end{proof}

This observation with application of lemma \ref{Schmidt1} can be effectively used to generate k-extendible states.

\begin{obs}
Let $\rho_{AB_{1}}$ be k-extendible to a pure symmetric state
$|\Psi\rangle_{AB_{1} \ldots B_{k+1}}$ then for ordered vectors of
eigenvalues of $\rho_{AB_{1}}$ and $\rho_{B_{2} \ldots B_{k+1}}$ there holds:
\begin{equation}\label{pure1}
\lambda^{\downarrow}(\rho_{AB_{1}})=\lambda^{\downarrow}(\rho_{B_{2} \ldots B_{k+1}})\;
\end{equation}

\end{obs}
\begin{proof}
The proof is immediate applying Schmidt decomposition and results of (\ref{Schmidt1}).
\end{proof}

\section{Symmetric extendibility of composite systems}

In this section we explore symmetric extendibility of complex systems consisting of n pairs.
All following statements are vital for protocols acting on multiple pairs of states.

For further results of the following section we will present a generalized version of a lemma \cite{MNPH} up to k-extendible
maps stating that no matter what operation Alice and Bob can perform,
the symmetric state shared between Alice and Bob will keep its symmetric extendibility.
The following lemma indicates a natural fact that one cannot produce k-extendible state from n-extendible state ($n>k$) by means of 1-LOCC $\Lambda_{\rightarrow}(\cdot)$ even if acts on any number
of pairs:

\begin{lem}
Let $\Lambda_{\rightarrow}$ be a 1-LOCC quantum operation (not necessarily trace-preserving):
\[
\Lambda_{\rightarrow}(\rho)=\sum_{ij}(I\otimes B_{ij})(A_{i}\otimes I)\rho(A_{i}\otimes I)^{\dag}(I\otimes B_{ij})^{\dag}
\]
where $\sum_{i}A_{i}A_{i}^{\dag}\leq I$ and $\sum_{j}B_{ij}B_{ij}^{\dag}= I$ for all i since Bob cannot
communicate the outcome of a probabilistic operation back to Alice.
If $\rho$ is k-extendible state then $\Lambda_{\rightarrow}(\rho)$ is n-extendible and $n\geq k$.
\end{lem}

One may state a non-trivial question if it is feasible to achieve symmetric extendibility of a composition
of quantum states when at least one of them is not-symmetric extendible. The result of this question
is crucial both for quantum security applications and measuring quantum entanglement. The following lemma
casts some light on this field:

\begin{lem}\label{additivity}
If $\rho_{AB}\in B(\cal{H}^N_A\otimes\cal{H}^M_B)$ is not
symmetrically extendible state then there does not exist any such
a state $\rho_{A'B'}\in B(\cal{H}^K_{A'}\otimes\cal{H}^L_{B'})$
that $\rho_{AB}\otimes \rho_{A'B'}$ would be symmetrically
extendible in respect to $BB'$ subsystem.
\end{lem}
\begin{proof}
Conversely, let $\rho_{ABA'B'}=\rho_{AB}\otimes \rho_{A'B'}$ be a
symmetrically extendible state acting on
$B(\cal{H}^N_A\otimes\cal{H}^M_B\otimes\cal{H}^K_{A'}\otimes\cal{H}^L_{B'})$.
Therefore, one notes that $\rho_{ABA'B'}$ after swapping to
$\rho_{AA'BB'}$ can be represented by method (\ref{rep1}) in an
appropriate basis including generators of group $SU(N)\otimes
SU(K)\otimes SU(M)\otimes SU(L)$ and further, can be extended to a
$1$-rank symmetric extension
$\rho_{AA'BB'\widetilde{B}\widetilde{B'}}$ where we extend $BB'$
part as follows:
\begin{eqnarray}
  \rho_{AA'BB'\widetilde{B}\widetilde{B'}} &=& \sum_{ijkl}
  \alpha_{ijkl}T_{ijklkl} + \\ &+&
  \sum_{ijklmn}\beta_{ijklmn}(T_{ijklmn}+T_{ijmnkl}) \nonumber
\end{eqnarray}
with tensors
$T_{ijklmn}=\sigma^{i}\otimes\sigma^{j}\otimes\sigma^{k}\otimes\sigma^{l}\otimes\sigma^{m}\otimes\sigma^{n}$.
Following we derive the state $\rho_{AB\widetilde{B}}$ of system
$AB\widetilde{B}$ tracing out that of $A'B'\widetilde{B'}$. For
the fact that
$Tr[\sigma^{i}\otimes\sigma^{j}\otimes\sigma^{k}]=Tr(\sigma^{i})Tr(\sigma^{j})Tr(\sigma^{k})$
and $Tr[\sigma^{i}]=\delta_{1i}$ after tracing out only elements
with $\sigma^{0}=I$ remain, namely, one obtains:
\begin{eqnarray}
  \rho_{AB\widetilde{B}} &=& \sum_{ik}
  \alpha_{i1k1}T_{i1k1k1} + \\ &+&
  \sum_{ikm}\beta_{i1k1m1}(T_{i1k1m1}+T_{i1m1k1}) \nonumber
\end{eqnarray}
Hence, $\rho_{AB\widetilde{B}}$ is $1$-rank symmetric extension of
$\rho_{AB}$ that is in contradiction with the assumption that the
latter is not symmetrically extendible.
\end{proof}

\begin{cor}
If $\rho_{AB}\in B(\cal{H}^N_A\otimes\cal{H}^M_B)$ is at most
k-extendible state then there does not exist any such
a state $\rho_{A'B'}\in B(\cal{H}^K_{A'}\otimes\cal{H}^L_{B'})$
that $\rho_{AB}\otimes \rho_{A'B'}$ would be k+1-
extendible in respect to $BB'$ subsystem.

\end{cor}

\begin{lem}\label{zero-way}
  Assume that $\rho_{AB}\in B(\cal{H}_A\otimes\cal{H}_B)$ is not symmetric extendible and there exists a local operation $\mathbb{F}$ acting on
  A-part such that $\sigma_{AB}=(\mathbb{F}\otimes id)\rho_{AB}(\mathbb{F}^{\dag}\otimes id)/Tr[(\mathbb{F}\otimes id)\rho_{AB}(\mathbb{F}^{\dag}\otimes id)]$ is a symmetric extendible state.

  Then for any \textit{local operations} $\mathbb{A}$ and $\mathbb{B}$ acting on
  A and B part of the system:
  \begin{equation}
  \mathbb{A}=U \left(
            \begin{array}{cccc}
              \alpha_{0} &  &  & \\
               & \alpha_{1} &  & \\
               &  & \ddots & \\
               &  &  & \alpha_{i}\\
            \end{array}
          \right) U^{\dag}
 \end{equation}
 \begin{equation}
  \Lambda(\rho_{AB})=\frac{\mathbb{A}\otimes\mathbb{B}\rho_{AB}\mathbb{A}^{\dag}\otimes\mathbb{B}^{\dag}}{Tr(\mathbb{A}\otimes\mathbb{B}\rho_{AB}\mathbb{A}^{\dag}\otimes\mathbb{B}^{\dag})}
 \end{equation}
  where for all i $0<\alpha_{i}\leq 1$ and U denotes an unitary operation ($\mathbb{B}$ has a corresponding structure),
  there exists a local operation $\widetilde{\mathbb{F}}$ such that
  $\widetilde{\sigma}_{AB}=(\widetilde{\mathbb{F}}\otimes id)\Lambda(\rho_{AB})(\widetilde{\mathbb{F}}^{\dag}\otimes id)/Tr[(\widetilde{\mathbb{F}}\otimes id)\Lambda(\rho_{AB})(\widetilde{\mathbb{F}}^{\dag}\otimes id)]$ is symmetric extendible and $\dim\mathbb{F}=\dim\mathbb{\widetilde{F}}$.

\end{lem}

\begin{proof}
To prove this lemma, it suffices to note that $\mathbb{A}=UDU^{\dag}$ with a diagonal matrix $D$. Further, we observe
that $\widetilde{\mathbb{F}}=\mathbb{F}\circ UD'U^{\dag}$ where $D'D=id$. The latter is possible due to the condition that for all $i$ there
holds: $0<\alpha_{i}\leq 1$ and we easily observe that $\mathbb{F}=\widetilde{\mathbb{F}}\circ\mathbb{A}$. This brings us to conclusion that
$\widetilde{\mathbb{F}}\Lambda(\rho_{AB}) \widetilde{\mathbb{F}}^{\dag}$ is a symmetric extendible operator (after normalization becoming a physical state).
For B-part the proof can be conducted in a similar manner as in particular, the local operation $\mathbb{B}$ is also revertible.
\end{proof}

\textit{Remark.} It casts some light on a fact that local operations actually does not change the amount of
symmetric extendibility embedded in a state.

This lemma is of a great importance for private security and entanglement distillation studies, as we can always build a
symmetric extension $\Gamma_{ABE}$ of a state $\widetilde{\sigma}_{AB}$ which means that Eve potentially has a state $\rho_{E}=\rho_{B}=Tr_{A}\widetilde{\sigma}_{AB}$ and operates on
such a space. To support this statement one can further derive the corollary about extendibility of any quantum state with a proposal of new extendible number of a quantum state:
\begin{defn}
For any $\rho_{AB}$,  $\eta_{SE}(\rho_{AB})=\max_{\mathbb{F}}\dim\mathbb{F}$ is called the extendible number of a state $\rho_{AB}$ where
$({\mathbb{F}\otimes id})\rho_{AB}({\mathbb{F}^{\dag}\otimes id})$ is a symmetric extendible operator and $\mathbb{F}$ is a local operation acting on A.
\end{defn}

\begin{cor}
Any state $\rho_{AB}\in B(\cal{H}_A\otimes\cal{H}_B)$ with extendible number $\eta_{SE}$ can be extended to a state $\rho_{ABE}\in B(\cal{H}_A\otimes\cal{H}_B\otimes\cal{H}_E) (\dim\cal{H}_B=\dim\cal{H}_E)$
where exists filtering operation $\mathbb{F}$ on A so that $({\mathbb{F}\otimes id})\rho_{ABE}({\mathbb{F}^{\dag}\otimes id})$ is invariant due to permutation of B and E.
\end{cor}

Naturally, there holds: if $\eta_{SE}(\rho_{AB})=rank(\rho_{A})$, then the state is symmetric extendible.

One may raise further a very important question how to create the property of
symmetric non-extendibility
both in case of single states and collective systems using only
local operations or additionally one-way communication that
naturally will have implications for distillability and capacities
of corresponding states and channels.
\begin{lem}
Let $\rho_{AB}\in B(\cal{H}_{AB})$ be a state possessing at most k-rank symmetric extension
where $k<\infty$ then there does not exist any 1-LOCC protocol represented by
$\Lambda_{A\rightarrow BC}:B(\cal{H}_{ABC})\rightarrow (\cal{\widetilde{H}}_{ABC})$
(not necessarily trace-preserving):
\begin{equation}
\Lambda_{A\rightarrow BC}(\rho_{AB}\otimes \sigma_{C})=\widetilde{\rho}_{ABC}
\end{equation}
so that $\widetilde{\rho}_{ABC}$ is a symmetric extension of $\rho_{AB}$ and
$\sigma_{C}\in B(\cal{H}_{C})$ is an additional resource on Bob's side.
\end{lem}
\begin{proof}
Since $\rho_{AB}$ is k-extendible, one can assume that its symmetric extension
is realized to $\rho_{ABB_{1}\ldots B_{k}}$ but $B_{1}\ldots B_{K}$-part
is possessed by Eve. Obviously no communication between Eve and Bob in such a scenario
is allowed so that Bob cannot detect locally Eve and further, since the set of symmetric extendible states is closed under $1-LOCC$ operations \cite{MNPH} even if
Alice and Bob had engaged one-way communication they cannot break symmetric extendability of
$\rho_{AB}$ and so cannot eliminate Eve if the symmetric extension had been realized.

Therefore, assuming that on the contrary $\Lambda_{A\rightarrow BC}$ enables creation of
a symmetric extension:
\begin{equation}
\Lambda_{A\rightarrow BC}\otimes id_{B_{1}\ldots B_{k}} (\rho_{ABB_{1}\ldots B_{k}}\otimes \sigma_{C}) = \Omega
\end{equation}
resulting state $\Omega$ would be k+1-symmetric extension of $\rho_{AB}$ that contradicts
the lemma's assumption about extendibility of this state and completes the proof.
\end{proof}

\textit{Remark.} The aforementioned statements holds as well in asymptotic regime
due to results of \ref{additivity} that can be extended for an infinite case.
\\

As a result of the above lemmas we can conclude that in general for creation of any
symmetric extension one needs to engage two-way communication.

\section{Symmetric extendible component in quantum states}

In this section we consider vulnerability of quantum states to the loss of non-symmetric extendibility property asking
how easily the quantum state becomes symmetric extendible by distraction of its sub-system or how much of symmetric extendibility can be extracted from
the state. When the former recalls lockability
of entanglement, the latter relates to the best symmetric approximation subject responding to the question: how much of non-symmetric extendible
component has to be mixed with symmetric extendible state so that it becomes non-symmetric extendible?

The general idea of locking a property of a quantum state
relates to the loss or decrease of this property subjected to a measurement or
discarding of one qubit. It has been shown \cite{KH, Winter1} that entanglement of formation,
entanglement cost and logarithmic negativity are lockable measures which manifests as an
arbitrary decrease of those measures after measuring one qubit.

Herewith, we analyze in fact locking of non-symmetric extendibility in sense that discarding one qubit from
the quantum state that is not symmetric extendible leads to the loss of this property. Further, we derive
implications for quantum security applying one-way communication between engaged parties Alice and Bob.

We shall show now that the property of non-symmetric extendibility of an arbitrary state $\rho_{AB}$
can be destroyed by measurement of one qubit and in result it presents how easily a quantum state
can be removed of one-way distillability and security.

Let us consider bipartite quantum state shared between Alice and Bob on the Hilbert
space $\cal{H}_{A}\otimes\cal{H}_{B}\cong \cal{C}^{d+2}\otimes\cal{C}^{d+2}$

\begin{equation}\label{state1}
\rho_{AB}=\frac{1}{2d-1} \left[
            \begin{array}{cccc}
              d P_{+} & 0 & 0 & \cal A\\
              0 & 0 & 0 & 0\\
              0 & 0 & 0 & 0\\
              \cal A^{\dagger} & 0 & 0 & \sigma\\
            \end{array}
          \right]
\end{equation}
where $P_{+}$ is a maximally entangled state on $\cal{C}^{d}\otimes\cal{C}^{d}$,  $\sigma=\sum^{d-1}_{i=1}|i\;0\rangle\langle i\;0|$ and $\cal A$ is an arbitrary chosen
operator so that $\rho_{AB}$ represents a correct quantum state.
This matrix is represented in the computational basis $|00\rangle, |01\rangle, |10\rangle, |11\rangle$
held by Alice and Bob and possess a singlet-like structure.
Whenever one party (Alice or Bob) measures the state in the local computational basis,
the state decoheres and off-diagonal elements vanish which leads to a symmetric extendible
state \cite{MNPH}:
\begin{equation}
\Upsilon_{AB}=\frac{d}{2d-1}P_{+}+\frac{1}{2d-1}\sum^{d-1}_{i=1}|i\;0\rangle\langle
i\;0|
\end{equation}
from which no entanglement nor secret key can be distilled by means of one-way communication and
local operations. Clearly this example shows that from a non-symmetric extendible state possessing
large entanglement cost and non-zero one-way secret key one can easily obtain a symmetric structure by
discarding small part of the whole system destroying possibility of entanglement distillation and secret
key generation by means of 1-LOCC.

Thus, it is interesting to consider how much of symmetric extendibility is embedded in a given state $\rho_{AB}$ as
it can be expected that the more symmetric extendibility is hidden in a state, the less vulnerable for losses of
one-way distillable entanglement and security it is. Besides analysis of symmetric structures in projected subspaces, we will also propose to perform this task by means
of \textit{the best symmetric extendible approximation} \cite{Lutk3, KarnasLew} that decomposes the state into
a symmetric extendible component $\sigma_{ext}$ and non-symmetric extendible component $\sigma_{next}$:
\begin{equation}
\rho_{AB}=\max_{\lambda}\lambda \sigma_{ext} + (1-\lambda)\sigma_{next}
\end{equation}
We denote by $\lambda_{max}(\rho)$ the maximum weight of extendibility \cite{Lutk3} of $\rho_{AB}$ where
$0 \leq \lambda_{max}(\rho)\leq 1$,thus, all symmetric extendible states have the weight
$\lambda_{max}=1$. It is proved in \cite{Lutk2, Lutk3} that in case of one-way protocols only
the non-symmetric extendible component can be effectively utilized for generation of a secret key and
it confirms that the notion of symmetric extendibility is crucial for consideration of one-way
entanglement and key distillation.

However, we show that there exist states which do not possess any symmetric extendible component in the
aforementioned decomposition but there can be a large symmetric extendible component embedded in them.
An example of such a state is given above (\ref{state1}) and one can derive the following statement about general
structure of such states:

\begin{lem}\label{lock1}
Consider a state $\gamma$ on $\cal{H}_{AA'BB'}=\cal{H}_{A}\otimes\cal{H}_{A'}\otimes\cal{H}_{B}\otimes\cal{H}_{B'}\sim \cal{C}^{d}\otimes\cal{C}^{d}\otimes\cal{C}^{d}\otimes\cal{C}^{d}$:
\begin{equation}
\gamma=\rho\otimes\sigma
\end{equation}
being a composition of an arbitrary chosen state
$\sigma \in B(\cal{H}_{A'}\otimes \cal{H}_{B'})$ and a non-symmetric extendible state
$\rho \in B(\cal{H}_{A}\otimes \cal{H}_{B})$ with no symmetric extendible component $\lambda_{max}(\rho)=0$.
Then for the best extendible approximation of $\gamma$ there holds $\lambda_{max}(\gamma)=0$, i.e. there is no
symmetric extendible component in $\gamma \in B(\cal{H}_{AA'BB'})$.
\end{lem}
\begin{proof}
Conversely, assume that there exists decomposition of $\gamma_{AA'BB'}$ with non-zero symmetric
extendible component, i.e. $\lambda\neq 0$:
\begin{equation}
\gamma_{AA'BB'}=\lambda \sigma_{ext}+(1-\lambda)\rho_{ne}
\end{equation}
then both components would be supported on $\cal{H}_{AA'BB'}$ and one
can search for a decomposition of $\gamma_{AA'BB'}$ after tracing out
A'B'-part. Due to additivity of a partial trace operation $\Gamma_{X}(\cdot)=Tr_{X}(\cdot)$
we obtain:
\begin{equation}
\Gamma_{A'B'}(\gamma_{AA'BB'})=\lambda \Gamma_{A'B'}(\sigma_{ext})+(1-\lambda)\Gamma_{A'B'}(\rho_{ne})
\end{equation}
\\
and, further, basing on a symmetric extendibility property of composite systems \cite{MNPH} one derives that
tracing out $A'B'$ from $\sigma_{ext}$ does not destroy its symmetric extendibility and produces symmetric
extendible state $\widetilde{\sigma}_{ext}$:
\begin{equation}
\rho=\lambda \widetilde{\sigma}_{ext}+(1-\lambda)\widetilde{\rho}_{ne}
\end{equation}
Thus, the initial assumption would imply existence of a non-zero symmetric extendible component of the state
$\rho$ that contradicts the aforementioned decomposition.
\end{proof}

Following one can make an immediate observation about any private quantum state \cite{KH1}:

\begin{cor}
Any private quantum state $\gamma_{ABA'B'} \in B(\cal{H}_{ABA'B'})$:
\begin{equation}
\gamma_{ABA'B'}=\frac{1}{2}\sum_{i,j=0}^{1}|ii\rangle\langle jj|\otimes U_{i}\rho_{A'B'}U_{j}^{\dagger}
\end{equation}
does not possess symmetric extendible component, i.e. $\lambda_{max}=0$.
\end{cor}
\textit{Remark.}
The proof is conducted in analogy to the proof of \ref{lock1} but this state represents a twisted composition of
singlet and an arbitrary chosen state $\sigma$ where AB-part is the key part of the state and is not symmetric extendible due to the observation that secure states cannot be symmetric extendible \cite{Lutk3}.

Basing on previous studies of entanglement measures and importance of symmetric extendible states, we
introduce the following best symmetric approximed entanglement monotone (as a counterpart of BSA in \cite{KarnasLew}):
\begin{prop} For any $\rho \in B(\cal{H}_A\otimes\cal{H}_B)$ having best symmetric decomposition $\rho_{AB}=\max_{\lambda}\lambda \sigma_{ext} + (1-\lambda)\sigma_{next}$,
the best symmetric approximated entanglement monotone is defined as:
\begin{equation}
E^{ss}(\rho)=1-\lambda_{max}(\rho)
\end{equation}
\end{prop}

\begin{proof}
(We will prove that the quantity meets necessary conditions to be an entanglement monotone.)
\\1. If $\rho$ is separable, i.e. also symmetric extendible, then $\lambda_{max}=1$ and $E^{ss}(\rho)=1-\lambda_{max}=0$.
\\2. $E^{ss}(\rho)$ is invariant under local unitary operations since application of local operations $U_{A}$ and $U_{B}$ on $\sigma_{ext}$ leaves it extendible
to the third part B', i.e. $E^{ss}(U_{A}\otimes U_{B}\rho U_{A}^{\dag}\otimes U_{B}^{\dag}) \geq E^{ss}(\rho)$ and vice versa.
\\3. For any local POVM $V_{i}$, there holds:
\begin{eqnarray*}
1-\lambda_{max}(\rho)&\geq& \sum_{i}(1-\lambda^{\max}_{i}(\rho_{i})Tr(V_{i}\rho V^{\dag}_{i}))\\
&\geq& \sum_{i}E^{ss}(\rho_{i})Tr(V_{i}\rho V^{\dag}_{i}))
\end{eqnarray*}
and $\rho_{i}=V_{i}\rho V^{\dag}_{i}/Tr(V_{i}\rho V^{\dag}_{i})$.
\end{proof}

It is interesting to notice that for two-qubit states on $\mathbb{C}^{2}\otimes \mathbb{C}^{2}$ there holds a non-trival observation about best symmetric approximated decomposition:
\begin{equation}\label{decompose}
\rho=\lambda\sigma_{ext} + (1-\lambda)|\Psi\rangle\langle\Psi|
\end{equation}
with $\sigma_{ext}$ being a symmetric extendible component that appears in $\rho$ with highest probability.
The proof of this observation can be based on BSA with separable components \cite{KarnasLew} where $\rho=\alpha\sigma_{sep} + (1-\alpha)|\Psi\rangle\langle\Psi|$.
As set of separable states is a subset of the convex set of symmetric extendible states, then for any dimension $\alpha \leq \lambda$. Further, due to the fact that any
two-qubit state has best separable decomposition into a separable and projective entangled component, we conclude that $\lambda\sigma_{ext}=\alpha\sigma_{sep} + \beta|\Psi\rangle\langle\Psi|$ for arbitrary chosen $\beta$.

These propositions can simplify potentially many research problems like analysis of $CHSH$ regions vs. symmetric extendibility of states \cite{TRudolph} represented in the steering ellipsoid formalism
or just further analysis on security and distillability of all $\mathbb{C}^{2}\otimes \mathbb{C}^{2}$  states.

Following the results of \cite{RQuesada}, one can immediately propose max-relative entropy monotone based
on this decomposition, i.e. $D_{\max}(\sigma\parallel\rho)\equiv\log\min\{\lambda:\sigma \leq \lambda\rho\}$ and $\textrm{supp}\sigma \subseteq \textrm{supp}\rho$ with max-relative entropy being interpreted as a probability of finding $\sigma$ in decompositions of $\rho$.
This leads immediately to
$\lambda=\max (2^{-D_{\max}(\sigma_{ext}\parallel \rho)})$.

An open question is: whether for one-way distillable entanglement we can state that $D_{\rightarrow}(\rho) \leq (1-\lambda_{max}(\rho))D_{\rightarrow}(\sigma_{next})$?

\section{Implications for one-way entanglement distillability and private key}

Studies on symmetric extendibility in a context of measures of entanglement like squashed entanglement \cite{Winter1,Winter2}, security of quantum protocols \cite{Lutk3}
and quantum maps gain a substantial interest. Recently a great attention has been paid to so called
k-extendible maps \cite{MPiani,Chiri, Brandao1} and recovery maps \cite{Fawzi,Brandao} where it is proved that small value of squashed entanglement implies closeness to highly
extendible states. These results show importance of symmetric extendibility notion for analysis of one-way quantum communication rates.
Inspired by these findings, we propose further an important conjecture about distillability of
all non-symmetric extendible states and analyze behavior of a secret key rate in a neighborhood of
symmetric extendible states.

Basing on theory of entanglement distillability we state the following conjecture in domain of
one-way communication linking it directly with symmetric extendibility of quantum states:
\begin{conj}\label{conj4}
Any state $\rho_{AB}$ on $\mathcal{H}=\mathcal{H}_{A}\otimes \mathcal{H}_{B}$ is
one-way distillable if and only if there exists a two-dimensional projector $P:\mathcal{H}^{n}_{A}\rightarrow \mathbb{C}^{2}$ such that for some $n\geq 1$ the state:
\begin{equation}
  \widetilde{\rho}_{AB}=(P\otimes id)\rho_{AB}^{\otimes n}(P\otimes id)^{\dag}
\end{equation}
is not symmetrically extendible.
\end{conj}
For a potential proof, it is an immediate observation that one-way distillable quantum states cannot be symmetric extendible \cite{MNPH}, yet it is an open question if there exists a two-qubit state that is not at the same time symmetric extendible nor one-way distillable.
Since we know conditions for symmetric extendibility of two-qubit states \cite{Lutk1, Lutk4}, this conjecture if true would simplify
analysis of entanglement of two-qubit states and capacity of channels acting on such spaces substantially. On the contrary, if there exist two-qubit states that are neither symmetric extendible nor one-way distillable then
they would be one-way counterparts of bound entangled states for two-way distillability in higer dimensions. An analysis of this subject
seems to be of a great importance for further studies on quantum secure protocols and structure of entanglement.

As an example, it is worth mentioning Werner states \cite{Werner} and the hypothesis about NPT bound entangled states \cite{DiVincenzo, Dur1}. The structure of the Werner states is as follows:
\begin{equation}
\rho_{W}(\alpha)=\frac{id+\alpha\mathbb{P}}{d^2+\alpha d}
\end{equation}
where $\mathbb{P}=\sum_{i,j=0}^{d-1}|ij\rangle\langle ji|$. The state is separable for $1 \geq\alpha\geq -\frac{1}{d}$, NPT for $-\frac{1}{d}>\alpha\geq\-1$ and two-way 1-distillable for $-\frac{1}{2}>\alpha\geq -1$. Applying the conditions for symmetric extendibility \cite{Lutk4}, we found that for $d=2$, the state is non-symmetric extendible for $-0.8 \geq \alpha \geq -1$.
We analyzed potential one-way distillability of the state for the region of non-symmetric extendible Werner states with non-positive coherent information, namely for $-0.8\geq\alpha>\cong -0.85559$. The latter condition excludes all those states being distilled by well-known one-way hashing protocol. The analysis was performed for two-copies of the state and over $10^{8}$ random filtering operations on Alice' side and random unitary operations on Bob's side. However, the protocol was not able to distill states with positive coherent information which suffices to distill entanglement with the hashing protocol.
Therefore, it is \textit{an open question} if  the state is one-way distillable in the region $-0.8\geq\alpha>\cong -0.85559$ or it is one-way 'bound entangled' which would be a counterpart of bound entanglement concept in two-way communication domain.

As all symmetric extendible state do not posses any private key, we can expect that in close neighborhood to the set of such states all other states
can have only a small amount of distillable private key. That would have to be true assuming at least local continuity of private key $K_{\rightarrow}(\cdot)$ in
such a neighborhood.
To anylyze this subject, we start reminding an important theorem about entropic inequalities for conditional entropies of sufficiently close states in terms of a trace norm:
\begin{thm}\cite{Alicki}
For any two states $\rho_{AB}$ and $\widetilde{\rho}_{AB}$ on $\cal{H}_{AB}=\cal{H}_{A}\otimes \cal{H}_{B}$,
let $\epsilon\equiv \parallel \rho_{AB} - \widetilde{\rho}_{AB} \parallel _{1}$ and let $d_{A}$ be the dimension of
$\cal{H}_{A}$, then the following estimate holds:
\begin{equation}\label{continuity}
|S(A|B)-S(\widetilde{A}|\widetilde{B})|\leq 4 \epsilon \log d_{A} + 2\eta(1-\epsilon) + 2\eta(\epsilon)
\end{equation}
In particular, the right hand side of ($\ref{continuity}$) does not explicitly depend on the dimension of $\cal{H}_{B}$.
\end{thm}

Further, to generate a secret key between Alice and Bob one can use \cite{Devetak05, Devetak06} a general tripartite pure
state $\rho_{ABE}$.
Alice engages a particular strategy to perform a quantum
measurement (POVM) described by $Q=(Q_{x})_{x \in \cal X}$ which
leads to: $\widetilde{\rho}_{ABE}=\sum_{x}|x\rangle\langle x|_{A}
\otimes Tr_{A}(\rho_{ABE}(Q_{x})\otimes I_{BE})$. Therefore,
starting from many copies of $\rho_{ABE}$ we obtain many copies of
cqq-states $\widetilde{\rho}_{ABE}$ and we restate the theorem
defining one-way secret key $K_{\rightarrow}$:

\textbf{Theorem 1.}\cite{Devetak05}
\textit{For every state $\rho_{ABE}$,
$K_{\rightarrow}(\rho) = \lim_{n\rightarrow\infty}\frac{K_{\rightarrow} ^{(1)}(\rho^{\otimes n})}{n}$,
with
$K_{\rightarrow}^{(1)}(\rho)=\max_{Q,T|X} I(X:B|T) - I (X:E|T)$
where the maximization is over all POVMs $Q=(Q_{x})_{x \in \cal X}$ and channels R
such that $T=R(X)$ and the information quantities refer to the state:
$\omega_{TABE}=\sum_{t,x} R(t|x)P(x)
|t\rangle\langle t|_{T}\otimes |x\rangle\langle x|_{A} \otimes
Tr_{A}(\rho_{ABE}(Q_{x})\otimes I_{BE}).$
The range of the measurement Q and the random variable T may be assumed to be bounded as follows:
 $|T|\leq d^{2}_{A}$ and $|\cal X|\leq d^{2}_{A}$ where T can be taken a (deterministic) function
 of $\cal X$.
}

Basing one the above results we will prove continuity of the quantity $K^{(1)}_{\rightarrow}(\rho)$ for one copy of a state $\rho$ and further, consider behavior of the measure in the asymptotic regime.

\begin{lem}
For any two states $\rho$ and $\widetilde{\rho}$ on $\cal{H}_{AB}=\cal{H}_{A}\otimes \cal{H}_{B}$,
let $\epsilon\equiv \parallel \rho - \widetilde{\rho} \parallel _{1}$ and let $d_{A}$ be the dimension of
$\cal{H}_{A}$, then the following estimate holds:

\begin{equation}\label{continuity2}
|K^{(1)}_{\rightarrow}(\rho)-K^{(1)}_{\rightarrow}(\widetilde{\rho})|\leq 8 \epsilon \log d_{A} + 4\eta(1-\epsilon) + 4\eta(\epsilon)
\end{equation}
\end{lem}
\begin{proof}
One can put for the quantity $K^{(1)}_{\rightarrow}(\rho)=S(BC)-S(ABC)-S(EC)+S(AEC)=-S(A|BC)+S(A|EC)$ and
respectively for $\widetilde{\rho}$ there holds
$K^{(1)}_{\rightarrow}(\widetilde{\rho})=-S(\widetilde{A}|\widetilde{BC})+S(\widetilde{A}|\widetilde{EC})$.
Further, engaging the results of (\ref{continuity})
it is easy to conduct the following implications for a chain of inequalities:
\begin{eqnarray*}
&&|K^{(1)}_{\rightarrow}(\rho) - K^{(1)}_{\rightarrow}(\widetilde{\rho})|=\\
&=& |[S(\widetilde{A}|\widetilde{BC})-S(A|BC)]+[S(A|EC)-S(\widetilde{A}|\widetilde{EC})]| \\
&\leq & |S(\widetilde{A}|\widetilde{BC})-S(A|BC)| + |S(A|EC)-S(\widetilde{A}|\widetilde{EC})| \\
&\leq &2[4 \epsilon \log d_{A} + 2\eta(1-\epsilon) + 2\eta(\epsilon)]
\end{eqnarray*}
\end{proof}

Since it is not possible to distill any secret key by means of one-way communication
and local operations from all symmetric extendible states, one can easily derive the following:
\begin{cor}
For any state $\rho$ on $\cal{H}_{AB}=\cal{H}_{A}\otimes \cal{H}_{B}$ being in distance $\epsilon$ to the nearest symmetric extendible
state $\widetilde{\sigma}$ in sense of a trace norm:
$\epsilon=\inf_{\sigma\in \Omega}\parallel \rho - \widetilde{\sigma} \parallel _{1}$ where
$\Omega$ denotes a convex set of symmetric extendible states on $\cal{H_{AB}}$, there holds:
\begin{equation}
K^{(1)}_{\rightarrow}(\rho)\leq 8 \epsilon \log d_{A} + 4\eta(1-\epsilon) + 4\eta(\epsilon)
\end{equation}
\end{cor}
\textbf{Example 3.}
As an example of application of the above corollary we will consider two states very close to one another
in sense of a trace norm $\parallel\cdot\parallel_{1}$ from which one is symmetric extendible and the another is non-symmetric extendible. This
shows that for one-copy applications the theorem can be used operationally to estimate one-way secret key rate
of quantum states. Following results of \cite{MNPH}, let us consider two arbitrary instances of a state on $\cal{H}_{AB}\cong \cal{C}^{d}\otimes\cal{C}^{d}$:
\begin{equation}
\Upsilon (\epsilon)=[\frac{d}{2d-1}+\epsilon/2]P_{+}+[\frac{1}{2d-1}-\frac{\epsilon}{2(d-1)}]\sum^{d-1}_{i=1}|i\;0\rangle\langle
i\;0|
\end{equation}
which is non-symmetric extendible for $\epsilon>0$. Namely, one can put into the inequality (\ref{continuity2}) two states $\Upsilon (\epsilon=0)$ and $\Upsilon (\epsilon>0)$. Since for all symmetric extendible states $\rho$ there holds:
$K^{(1)}_{\rightarrow}(\rho)=0$, then:
\[K^{(1)}_{\rightarrow}(\Upsilon(\epsilon>0))\leq 8 \epsilon \log d_{A} + 4\eta(1-\epsilon) + 4\eta(\epsilon).\]
where $\epsilon\leq\frac{2(d_{A}-1)}{2d_{A}-1}$.

It is proved \cite{Gvidal} that in any open set of distillable states, all asymptotic entanglement measures $E(\rho)$ are continuous as a function of a single copy of $\rho$, even though they quantify the entanglement properties of $\rho^{\otimes N}$ in the large $N$ limit. \\
However, the aforementioned theorem does not cast any light on the behavior of function
$K_{\rightarrow}(\cdot)$ on the boundary of
a set of all one-way distillable states adjacent to symmetric extendible states just due to the open conjecture \ref{conj4}. Motivated by this
insight we put an open question in the following form for $\epsilon$-neighborhood of symmetric extendible states
having zero one-way secret key rate:

\begin{conj}
For any state $\rho$ on $\cal{H}_{AB}=\cal{H}_{A}\otimes \cal{H}_{B}$ being in distance $\epsilon$ to the nearest symmetric extendible
state $\widetilde{\sigma}$ in sense of a trace norm:
$\epsilon=\inf_{\sigma\in \Omega}\parallel \rho - \widetilde{\sigma} \parallel _{1}$ where
$\Omega$ denotes a convex set of symmetric extendible states on $\cal{H_{AB}}$, there holds:
\begin{equation}
K_{\rightarrow}(\rho)\leq 8 \epsilon \log d_{A} + 4\eta(1-\epsilon) + 4\eta(\epsilon)
\end{equation}

\end{conj}

\section{Conclusions}
The theory of symmetric extendible states being crucial for analysis of one-way distillability and security of quantum states has
still many unsolved problems. In this paper we introduced some new concepts
related to classification of all symmetric extendible states and
analyzed mainly composite systems including also a symmetric extendible part. In section II. we introduced some observations about the general structure of symmetric extendible states. In section III. we analyzed the structure of composite systems where
its part is symmetric extendible and answered a general question of further extendibility of $k$-extendible states. We introduced a new notion of the extendible number of a quantum state that can be used in further studies on characterization of such states.
As presented in the paper, beside analysis of best symmetric extendible decompositions it might be very useful to
analyze a maximal symmetric extendible state that can be achieved by filtering on Alice' side.
Further, there has been a new one-way monotone based on the best symmetric approximation of quantum state introduced in section IV. treating about the symmetric
extendible component embedded in quantum states.

Finally, in section V. we studied also behavior of private key in neighborhood of symmetric extendible states showing that for one-copy a quantum state close to symmetric extendible state can possess only a small amount of private key.
One of the most intriguing open question relates to the conjecture about
one-way distillability of all two-qubit states which are not symmetric extendible. In consequence, that would simplify substantially full characterization
of two-qubit states in terms of their privacy and distillability. In relation to this question we analyzed Werner states
in the domain of non-positive coherent information which would indicate one-way NPT bound entangled features in case the conjecture
is not true.

\section{Acknowledgments}
The author thanks Pawel Horodecki for fruitful discussions and critical comments on this paper. This work is supported by the ERC
grant QOLAPS. Part of this work was performed at the National Quantum Information Center in Gdansk.

\end{document}